\documentclass{IEEEtaes}

\usepackage{color,array,amsthm}
\usepackage{graphicx}
\usepackage{amsmath}
\usepackage{amssymb}
\usepackage{booktabs}
\usepackage{multirow}
\usepackage{algorithmic}
\usepackage{algorithm}

\jvol{XX}
\jnum{XX}
\jmonth{XXXXX}
\paper{1234567}
\pubyear{2020}
\doiinfo{TAES.2020.Doi Number}

\DeclareMathOperator*{\argmin}{arg\,min}
\newtheorem{theorem}{Theorem}

\begin{document}

\title{Decentralized Bi-Level Optimization for Cooperative Orbital Transportation of a Rigid Payload}

\author{Sidhdharth D. Sikka}
\affil{Purdue University, West Lafayette, IN, 47907}

\author{Revathy Ravuri}
\affil{University of Manchester, Manchester, United Kingdom, M13 9PL}

\author{Hei Yin Jowett Chan}
\affil{Manifold Research Group, San Francisco, CA, 94110}

\author{Tim Echtermeyer}
\affil{University of Manchester, Manchester, United Kingdom, M13 9PL}

\author{Shaoshuai Mou}
\affil{Purdue University, West Lafayette, IN, 47907}

\receiveddate{}

\corresp{{\itshape (Corresponding author: Sidhdharth D. Sikka)}}

\authoraddress{Sidhdharth D. Sikka is with Purdue University, West Lafayette, IN 47907 USA, and also with the Manifold Research Group, San Francisco, CA 94110 USA (e-mail: \href{mailto:sikkas@purdue.edu}{sikkas@purdue.edu}). Revathy Ravuri is with the University of Manchester, Manchester M13 9PL, United Kingdom, and also with the Manifold Research Group, San Francisco, CA 94110 USA. Jowett Chan is with the Manifold Research Group, San Francisco, CA 94110 USA. Tim Echtermeyer is with the University of Manchester, Manchester M13 9PL, United Kingdom. Shaoshuai Mou is with Purdue University, West Lafayette, IN 47907 USA.}

\markboth{SIKKA ET AL.}{DECENTRALIZED BI-LEVEL OPTIMIZATION FOR COOPERATIVE ORBITAL TRANSPORTATION}
\maketitle

\begin{abstract}The rapid growth of the space economy, alongside increasing orbital congestion and debris, is driving demand for reliable on-orbit transportation of large and non-cooperative objects. Performing such transportation with a swarm of small spacecraft distributes risk and removes the single point of failure inherent to monolithic architectures, but requires solving a non-convex trajectory optimization coupling every agent's cone-constrained thrust to the payload's rotational dynamics. Existing approaches span decentralized control laws with asymptotic guarantees but limited constraint enforcement, and centralized and partially centralized planners that reintroduce a single point of failure. In this work, we present a decentralized bi-level framework driven by exact gradients: the problem is split into an inner convex thrust allocation and an outer search over payload torque trajectories. This decomposition yields fuel-efficient solutions quickly, without requiring full NLP convergence, and at short timescales its performance scales with the number of agents. Theoretical convergence guarantees are established, and the method is validated through extensive Monte Carlo simulation, including closed-loop recovery from in-flight faults. These properties make the framework well suited to spacecraft swarm applications.
\end{abstract}


\section{INTRODUCTION}
The rapid growth of the space economy has significantly increased demand for reliable and efficient on-orbit transportation solutions, to support both commercial and scientific applications. At the same time, increasing orbital congestion and the proliferation of large debris pose growing risks to sustainable space operations. Addressing these challenges requires transportation technologies capable of safely and efficiently altering the orbit of non-cooperative objects, whose tumbling motion and the limited thrust of servicing spacecraft already make even close-proximity operations demanding \cite{wang_safe_2025}. 

Monolithic mission architectures, with a single spacecraft performing the entire manoeuvre, are costly, risk prone, and lack redundancy. Multi-agent architectures offer a promising alternative, distributing the risk across multiple spacecraft, thereby also increasing robustness and operational flexibility through coordinated action.

While cooperative object transportation with multiple agents has been widely studied across terrestrial and aerial domains \cite{farivarnejad_multirobot_2022,leitner_multi-robot_2009,stuber_lets_2020,abeshtan_robust_2024}, including decentralized learning-based \cite{zhang_decentralized_2020} and model-predictive \cite{muhammed_realtime_2026} controllers for multi-robot payload transport, its application in orbital settings is comparatively limited. Notable approaches include decentralized PD control \cite{farivarnejad_fully_2021}, occlusion-based heuristics \cite{chen_occlusion-based_2015}, back-stepping control for space manipulators \cite{sucan_advantages_2011}, and cooperative multi-satellite detumbling of tumbling targets through designed torque control laws \cite{xu_cooperative_2025}. Formation-based methods have also been proposed, such as those by Meng et al. \cite{meng_decentralized_2018}, re-framing cooperative transport as a formation tracking problem. However, these methods typically offer asymptotic guarantees with limited constraint enforcement and little attention to fuel efficiency or convergence time; factors that are critical to real-time orbital operations.  

Partially centralized architectures have been explored for cooperative transportation tasks in orbit. In particular, mothership-style systems, featuring centralized planning and distributed actuation, have been proposed for missions such as asteroid exploration \cite{vance_evaluation_2020}, telescope assembly \cite{underwood_using_2015}, and orbital debris mitigation \cite{udrea_cooperative_2015}. Most recently, cooperative transportation of passive free-flyers by pushing chaser spacecraft has been demonstrated under human-in-the-loop velocity commands, with stability and safety enforced by barrier conditions and thrust allocated by a centralized mixed-integer program \cite{marchesini_multiagent_2026}. While effective in many scenarios, these architectures concentrate planning or allocation in a single node and therefore remain vulnerable to single-point failures that can jeopardize overall mission success.

Fully decentralized architectures remove this vulnerability by distributing computation and decision-making across the agents, and are increasingly seen as essential as space operations scale up \cite{asri_introductory_2024,adams_swarm_2026}. A mature body of distributed control delivers decentralized transport with force distribution and fault response \cite{meng_decentralized_2018} and fault-tolerant formation flying under actuator loss \cite{ma_adaptive_2023,santos_reconfigurable_2025,azadmanesh_decoupling_2026}. However, these methods guarantee asymptotic tracking or synchronization, treats thrust-direction cones and terminal boundary conditions only approximately, and optimize no mission objective such as fuel. The nearest decentralized \emph{optimizer} \cite{chang_decentralized_2026} uses duality to handle collision constraints between free-flying vehicles, not an objective coupled through a shared payload. Decentralized trajectory optimization with hard constraint satisfaction therefore remains open.

This paper addresses this problem of cooperative orbital transportation involving a rigid body and $N$ spacecraft agents that must collaboratively manoeuvre the object from an initial state to a desired final state while minimizing total fuel. Each agent is affixed to a fixed location on the payload’s surface and can exert thrust within a constrained angular range.

A constrained optimization framework is well-suited to this problem, as it naturally incorporates relative orbital dynamics, attitude dynamics, actuation constraints, and trajectory boundary conditions. Optimization-based planning has been applied to spacecraft swarms through sequential convex programming \cite{morgan_model_2014}, reachable-set collision constraints \cite{cui_trajectory_2024}, and hierarchical planning over dynamic communication topologies \cite{liu_hierarchical_2026}, but in centralized or hierarchical settings and for free-flying agents rather than agents rigidly coupled through a shared payload. However, the problem is inherently non-convex due to the coupling between agents’ control inputs and the rotational dynamics of the rigid body. Previous work \cite{sikka_orbit_2025} addressed this challenge using a centralized bi-level optimization approach where the overall problem was split into an inner convex problem that solved for optimal control input trajectory given a torque profile, while the outer non-convex optimisation problem used population-based heuristics to optimize the torque trajectory. This method is still susceptible to a single-point failure due to the centralised approach taken. 

In this paper, we extend the previous approach to a decentralized setting, taking advantage of the multi-agent architecture and thereby removing the single point of failure. Rather than treating the outer torque search as a purely heuristic problem, we observe that the optimal value of the inner convex problem is differentiable in the torque trajectory, and that its exact gradient is available at relatively low computational cost once the inner solver has converged. The outer problem is then solved by projected gradient descent with random restarts, and decentralization acquires a precise role: each agent runs an independent restart stream from its own random seeds, so a swarm of $N$ agents draws $N$ times as many restarts per unit wall-clock and can therefore reach basins of attraction a single node would miss, and a single max-consensus round over a line-of-sight communication graph delivers the swarm-wide best trajectory to every agent. The same decentralized search is embedded in a closed-loop recovery mode that detects in-flight actuation and communication faults, detumbles the payload, identifies and removes failed agents, and replans with the survivors.

The key contributions of this paper are:
\begin{enumerate}
  \item a first-order method for the outer problem of the bi-level formulation introduced in~\cite{sikka_orbit_2025}, using the exact envelope-theorem gradient of the inner optimal value;
  \item a decentralized multi-start architecture whose solution quality improves with the number of agents; and
  \item a closed-loop fault-recovery mode, made possible by the decentralization of the architecture.
\end{enumerate}

The remainder of the paper is organised as follows. Section~\ref{sec:formulation} describes the problem of interest, its bi-level reformulation, and the equivalence of the two. Section~\ref{sec:proposed} presents the proposed method and establishes its convergence guarantee. Section~\ref{sec:experiments} evaluates the method through Monte Carlo testing against sampling-based and direct-NLP baselines across cone half-angles, compute budgets, and fleet sizes, including closed-loop fault-survivability statistics, and Section~\ref{sec:results} reports the results. Section~\ref{sec:conclusion} discusses implications and future work.

\begin{figure}
    \centering
    \includegraphics[width=\linewidth]{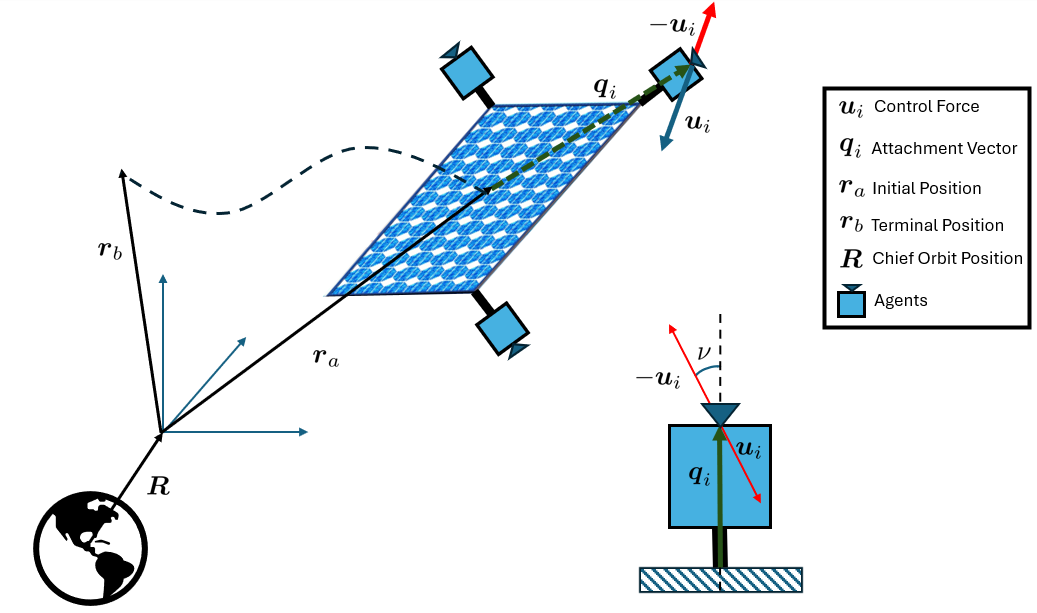}
    \caption{Agents Attached to Rigid Payload. Agents are depicted as blue squares. On the left of the figure, the center of the payload has initial position in the relative orbital frame denoted by $\boldsymbol{r}_a$ and terminal position in the relative orbital frame denoted by $\boldsymbol{r}_b$. The attachment vector $\boldsymbol{q}_i$ represents the vector pointing from the center of the payload to the agent. The Maximum Thrust Angle Constraint is depicted on the right of the figure, the control thrust applied by each agent must be aligned with the inward attachment direction $-\boldsymbol{q}_i$ within an acceptable angle $\nu$, so the agent pushes against the payload and its exhaust plume points away from it.}
    \label{fig:mainaim}
\end{figure}

\section{PROBLEM FORMULATION}
\label{sec:formulation}

\subsection{State, Input, and System Dynamics}

The system is formulated within the Hill frame, a relative orbital frame centered on a reference chief orbit. The payload rigid-body state is
\begin{equation}
  \boldsymbol{x}_p(t) = \mathrm{col}\{\boldsymbol{r}(t),\,\boldsymbol{v}(t),\,\boldsymbol{\phi}(t),\,\boldsymbol{\omega}(t)\}
  \in \mathbb{R}^9 \times \mathfrak{so}(3),
\end{equation}
where $\boldsymbol{r},\boldsymbol{v}\in\mathbb{R}^3$ are the relative position and velocity, $\boldsymbol{\phi}\in\mathbb{R}^3$ is the rotation vector parametrising attitude on the Lie algebra $\mathfrak{so}(3)$, and $\boldsymbol{\omega}\in\mathbb{R}^3$ is the angular velocity. Attitude is carried as a rotation matrix $R\in SO(3)$ and updated via the exponential map to eliminate norm drift. Each agent $i\in\{1,\ldots,N\}$ is rigidly mounted at a body-frame attachment point $\boldsymbol{q}_i^{\mathcal{B}}\in\mathbb{R}^3$; the Hill-frame attachment vector is $\boldsymbol{q}_i(t)=R(t)\,\boldsymbol{q}_i^{\mathcal{B}}$. The full state is
\begin{equation}
  \boldsymbol{x}(t) = \mathrm{col}\{\boldsymbol{x}_p(t),\,\boldsymbol{q}_1(t),\ldots,\boldsymbol{q}_N(t)\}
  \in \mathbb{R}^{9+3N} \times SO(3),
\end{equation}
and the stacked control input is $\boldsymbol{u}(t)=\mathrm{col}\{\boldsymbol{u}_1(t),\ldots,\boldsymbol{u}_N(t)\}\in\mathbb{R}^{3N}$, where $\boldsymbol{u}_i\in\mathbb{R}^3$ is the thrust vector of agent $i$.

The continuous-time dynamics are
\begin{equation}
  \dot{\boldsymbol{x}}(t) = \mathbf{A}(\boldsymbol{\omega},t)\boldsymbol{x}(t)
  + \mathbf{B}(\boldsymbol{q}_{1:N})\boldsymbol{u}(t),
\end{equation}
where $\mathbf{A}$ contains the linearised relative orbital dynamics (the time-varying Tschauner--Hempel equations, valid for an eccentric reference orbit) and the Euler rotational equations, and $\mathbf{B}$ maps individual thrust vectors to translational accelerations and body torques via the skew-symmetric operator $[\cdot]_\times$. The discrete-time update at step size $\Delta$ is
\begin{subequations}
\label{eq:discrete}
\begin{align}
  \begin{bmatrix}\boldsymbol{r}_{k+1}\\\boldsymbol{v}_{k+1}\end{bmatrix}
  &= \begin{bmatrix}\boldsymbol{r}_k\\\boldsymbol{v}_k\end{bmatrix}
  + \Delta\!\left(\boldsymbol{\Psi}(t_k)\begin{bmatrix}\boldsymbol{r}_k\\\boldsymbol{v}_k\end{bmatrix}
  + \frac{1}{m}\begin{bmatrix}\mathbf{0}\\\mathbf{I}_3\end{bmatrix}
  \sum_{i=1}^{N}\boldsymbol{u}_{i,k}\right),
  \label{eq:disc_trans}\\
  R_{k+1} &= R_k\cdot\exp\!\bigl(\Delta[\boldsymbol{\omega}_k]_\times\bigr),\quad
  R_0 = \exp([\boldsymbol{\phi}_0]_\times),
  \label{eq:disc_att}\\
  \boldsymbol{\omega}_{k+1} &= \boldsymbol{\omega}_k
  + \Delta\bigl(-\mathbf{J}^{-1}[\mathbf{J}\boldsymbol{\omega}_k]_\times\boldsymbol{\omega}_k
  + \mathbf{J}^{-1}\boldsymbol{\tau}_k\bigr),
  \label{eq:disc_omega}\\
  \boldsymbol{q}_{i,k} &= R_k\,\boldsymbol{q}_i^{\mathcal{B}},\quad\forall\,i.
  \label{eq:disc_q}
\end{align}
\end{subequations}
Equations~\eqref{eq:disc_trans}--\eqref{eq:disc_q} define the discrete-time map $\boldsymbol{x}_{k+1}=f_d(\boldsymbol{x}_k,\boldsymbol{u}_k)$. Equation~\eqref{eq:disc_att} propagates attitude directly on $SO(3)$, with the rotation vector recovered as $\boldsymbol{\phi}_k=\log(R_k)^\vee$; boundary attitudes are specified as rotation vectors $\boldsymbol{\phi}_a,\boldsymbol{\phi}_b\in\mathfrak{so}(3)$. The inertia tensor $\mathbf{J}$ is treated as constant since the payload mass dominates that of the agents.

\subsection{Optimal Control Problem}

Each agent is constrained to thrust within a cone of half-angle $\nu$ centred on its inward attachment direction $-\boldsymbol{q}_{i,k}$, so that thrust presses against the payload while the exhaust plume points away from it. Let $\mathbf{U}=\{\boldsymbol{u}_{i,k}\}\in\mathbb{R}^{3NT}$ denote the full thrust sequence. The problem of interest is to find the minimum-fuel thrust sequence
\begin{subequations}
\label{eq:ocp}
\begin{align}
  \argmin_{\mathbf{U}}&\;J(\mathbf{U})
  = \sum_{i=1}^{N}\sum_{k=0}^{T-1}\|\boldsymbol{u}_{i,k}\|
  \label{eq:ocp_obj}\\
  \text{s.t.}&\;\boldsymbol{x}_{k+1} = f_d(\boldsymbol{x}_k,\boldsymbol{u}_k),
  \label{eq:ocp_dyn}\\
  &\;\frac{-\boldsymbol{u}_{i,k}^\top\boldsymbol{q}_{i,k}}
        {\|\boldsymbol{u}_{i,k}\|\|\boldsymbol{q}_{i,k}\|}\geq\cos\nu,
        \quad\forall\,i,k,
  \label{eq:ocp_cone}\\
  &\;\boldsymbol{x}_0=\boldsymbol{a},\quad\boldsymbol{x}_T=\boldsymbol{b}.
  \label{eq:ocp_bc}
\end{align}
\end{subequations}
Problem~\eqref{eq:ocp} is non-convex due to the nonlinear attitude dynamics and the thrust-cone constraint~\eqref{eq:ocp_cone}.

\subsection{Bi-Level Reformulation}
\label{sec:reform}

The control input is decomposed by separating the net force and net torque:
\begin{equation}
  \mathbf{B}\boldsymbol{u}(t) = \tilde{\mathbf{B}}
  \begin{bmatrix}\displaystyle\sum_{i=1}^{N}\boldsymbol{u}_i(t)\\[4pt]\boldsymbol{\tau}(t)\end{bmatrix},
  \quad
  \boldsymbol{\tau}(t)=\sum_{i=1}^{N}[\boldsymbol{q}_i(t)]_\times\boldsymbol{u}_i(t),
\end{equation}
where $\tilde{\mathbf{B}}$ maps net force to translational acceleration and net torque through $\mathbf{J}^{-1}$ to angular acceleration. Once the initial condition and a torque trajectory $\boldsymbol{\tau}_{0:T-1}$ are fixed, Eqs.~\eqref{eq:disc_att}--\eqref{eq:disc_q} propagate all rotational quantities $\{R_k,\boldsymbol{\omega}_k,\boldsymbol{q}_{i,k}\}$ independently of the translational states. The remaining optimisation over $\mathbf{U}$ and $\{\boldsymbol{r}_k,\boldsymbol{v}_k\}$ is then convex. This motivates the bi-level decomposition:

\paragraph{Inner convex problem.} Given fixed $\boldsymbol{\tau}_{0:T-1}$:
\begin{subequations}
\label{eq:inner}
\begin{align}
  \argmin_{\mathbf{U}}&\;\sum_{i}\sum_{k}\|\boldsymbol{u}_{i,k}\|
  \label{eq:inner_obj}\\
  \text{s.t.}\;\;&\text{translational dynamics}~\eqref{eq:disc_trans},
  \label{eq:inner_dyn}\\
  &\boldsymbol{\tau}_k = \sum_i[\boldsymbol{q}_{i,k}]_\times\boldsymbol{u}_{i,k},
  \label{eq:inner_torque}\\
  &\frac{-\boldsymbol{u}_{i,k}^\top\boldsymbol{q}_{i,k}}
       {\|\boldsymbol{u}_{i,k}\|\|\boldsymbol{q}_{i,k}\|}\geq\cos\nu,
  \label{eq:inner_cone}\\
  &\boldsymbol{r}_0=\boldsymbol{r}_a,\;\boldsymbol{r}_T=\boldsymbol{r}_b,\;
   \boldsymbol{v}_0=\boldsymbol{v}_a,\;\boldsymbol{v}_T=\boldsymbol{v}_b.
  \label{eq:inner_bc}
\end{align}
\end{subequations}
The torque trajectory enters only through the constraint~\eqref{eq:inner_torque}. The solution $\xi(\boldsymbol{\tau}_{0:T-1})\in\mathbb{R}^{3NT}$ is the minimum-fuel thrust sequence consistent with the given torque trajectory.

\paragraph{Outer search problem.} Find the torque trajectory minimising the inner cost:
\begin{subequations}
\label{eq:outer}
\begin{align}
  \argmin_{\boldsymbol{\tau}_{0:T-1}}&\;L(\boldsymbol{\tau}_{0:T-1})
  = J\!\left(\xi(\boldsymbol{\tau}_{0:T-1})\right)
  \label{eq:outer_obj}\\
  \text{s.t.}\;\;
  &R_{k+1} = R_k\cdot\exp\!\bigl(\Delta[\boldsymbol{\omega}_k]_\times\bigr),
  \label{eq:outer_dyn_att}\\
  &\boldsymbol{\omega}_{k+1} = \boldsymbol{\omega}_k
  + \Delta\bigl(-\mathbf{J}^{-1}[\mathbf{J}\boldsymbol{\omega}_k]_\times\boldsymbol{\omega}_k
  + \mathbf{J}^{-1}\boldsymbol{\tau}_k\bigr),
  \label{eq:outer_dyn_omega}\\
  &R_0=\exp([\boldsymbol{\phi}_a]_\times),\quad R_T=\exp([\boldsymbol{\phi}_b]_\times),
  \label{eq:outer_att}\\
  &\boldsymbol{\omega}_0=\boldsymbol{\omega}_a,\quad
   \boldsymbol{\omega}_T=\boldsymbol{\omega}_b.
  \label{eq:outer_omega}
\end{align}
\end{subequations}

\subsection{Equivalence of the Bi-Level Reformulation}
\label{sec:equiv}

We now establish that the bi-level reformulation achieves the same optimal value as the original problem~\eqref{eq:ocp}. Let $\mathcal{F}$ denote the feasible set of~\eqref{eq:ocp}, that is, all thrust sequences satisfying~\eqref{eq:ocp_dyn}--\eqref{eq:ocp_bc}, and let $J^\star\triangleq\min_{\mathbf{U}\in\mathcal{F}}J(\mathbf{U})$ be its optimal value. Define the induced torque trajectory of a control sequence $\mathbf{U}$ as
\begin{equation}
  \boldsymbol{\tau}_k(\mathbf{U}) \triangleq \sum_{i=1}^{N}[\boldsymbol{q}_{i,k}]_\times\boldsymbol{u}_{i,k},
  \quad k=0,\ldots,T-1,
\end{equation}
and the set of admissible torque trajectories satisfying the rotational boundary conditions as
\begin{equation}
  \mathcal{T} \triangleq \bigl\{\boldsymbol{\tau}_{0:T-1} :
  \boldsymbol{\phi}_0=\boldsymbol{\phi}_a,\;
  \boldsymbol{\phi}_T=\boldsymbol{\phi}_b,\;
  \boldsymbol{\omega}_0=\boldsymbol{\omega}_a,\;
  \boldsymbol{\omega}_T=\boldsymbol{\omega}_b \bigr\},
\end{equation}
where the boundary values are evaluated after propagating~\eqref{eq:disc_att}--\eqref{eq:disc_omega} under $\boldsymbol{\tau}_{0:T-1}$. For each $\boldsymbol{\tau}\in\mathcal{T}$, the inner feasible set is
\begin{align}
  \mathcal{U}(\boldsymbol{\tau}) \triangleq \Bigl\{\mathbf{U} :\;&
  \text{translational dynamics~\eqref{eq:disc_trans}},\notag\\
  &\boldsymbol{\tau}_k = \textstyle\sum_i[\boldsymbol{q}_{i,k}]_\times\boldsymbol{u}_{i,k},\quad\forall k,\notag\\
  &\tfrac{-\boldsymbol{u}_{i,k}^\top\boldsymbol{q}_{i,k}}{\|\boldsymbol{u}_{i,k}\|\|\boldsymbol{q}_{i,k}\|}\geq\cos\nu,\quad\forall i,k,\notag\\
  &\boldsymbol{r}_0=\boldsymbol{r}_a,\;\boldsymbol{r}_T=\boldsymbol{r}_b,\;
   \boldsymbol{v}_0=\boldsymbol{v}_a,\;\boldsymbol{v}_T=\boldsymbol{v}_b\Bigr\}.
\end{align}
The outer objective of~\eqref{eq:outer_obj} is then the value function of the inner problem, and the bi-level optimal value is
\begin{equation}
  L(\boldsymbol{\tau})=\min_{\mathbf{U}\in\mathcal{U}(\boldsymbol{\tau})}J(\mathbf{U}),
  \qquad
  \widetilde{J}^\star\triangleq\min_{\boldsymbol{\tau}\in\mathcal{T}}L(\boldsymbol{\tau}).
\end{equation}

\textit{Claim:} $\widetilde{J}^\star = J^\star$.

\textit{Proof ($\widetilde{J}^\star \leq J^\star$).} Let $\mathbf{U}\in\mathcal{F}$ be any feasible solution of~\eqref{eq:ocp} and set $\boldsymbol{\tau}=\boldsymbol{\tau}(\mathbf{U})$. Since $\mathbf{U}$ satisfies the full dynamics of~\eqref{eq:ocp}, the induced torque trajectory satisfies all rotational boundary conditions, so $\boldsymbol{\tau}\in\mathcal{T}$. By construction, $\mathbf{U}$ also satisfies all constraints in $\mathcal{U}(\boldsymbol{\tau})$, hence $\mathbf{U}\in\mathcal{U}(\boldsymbol{\tau})$ and $L(\boldsymbol{\tau})\leq J(\mathbf{U})$. Taking the minimum over all $\mathbf{U}\in\mathcal{F}$:
\begin{equation}
  \widetilde{J}^\star = \min_{\boldsymbol{\tau}\in\mathcal{T}}L(\boldsymbol{\tau}) \leq J^\star.
\end{equation}

\textit{Proof ($J^\star \leq \widetilde{J}^\star$).} Let $\boldsymbol{\tau}\in\mathcal{T}$ and $\mathbf{U}\in\mathcal{U}(\boldsymbol{\tau})$. Since $\boldsymbol{\tau}\in\mathcal{T}$, propagating~\eqref{eq:disc_att}--\eqref{eq:disc_omega} yields rotational trajectories satisfying $\boldsymbol{\phi}_0=\boldsymbol{\phi}_a$, $\boldsymbol{\phi}_T=\boldsymbol{\phi}_b$, $\boldsymbol{\omega}_0=\boldsymbol{\omega}_a$, $\boldsymbol{\omega}_T=\boldsymbol{\omega}_b$. Since $\mathbf{U}\in\mathcal{U}(\boldsymbol{\tau})$, the translational boundary conditions hold, the cone constraints are satisfied, and the torque consistency $\boldsymbol{\tau}_k=\sum_i[\boldsymbol{q}_{i,k}]_\times\boldsymbol{u}_{i,k}$ holds at every step. Together, $(\mathbf{U},\boldsymbol{\tau})$ satisfies all constraints of~\eqref{eq:ocp}, so $\mathbf{U}\in\mathcal{F}$ and $J^\star\leq J(\mathbf{U})$. Taking the minimum over $\mathbf{U}\in\mathcal{U}(\boldsymbol{\tau})$ then over $\boldsymbol{\tau}\in\mathcal{T}$:
\begin{equation}
  J^\star \leq \min_{\boldsymbol{\tau}\in\mathcal{T}}L(\boldsymbol{\tau}) = \widetilde{J}^\star.
\end{equation}

Combining both directions gives $J^\star=\widetilde{J}^\star$. Furthermore, any optimizer $\mathbf{U}^\star$ of~\eqref{eq:ocp} induces an outer optimizer $\boldsymbol{\tau}^\star=\boldsymbol{\tau}(\mathbf{U}^\star)$, and conversely any outer optimizer $\boldsymbol{\tau}^\star$ paired with its inner minimizer $\mathbf{U}^\star\in\arg\min_{\mathbf{U}\in\mathcal{U}(\boldsymbol{\tau}^\star)}J(\mathbf{U})$ is an optimizer of~\eqref{eq:ocp}.

\textit{Remark 1 (principal chart).} The rotation-vector parametrisation $\boldsymbol{\phi}=\log(R)^\vee$ is used on its principal chart, $\|\boldsymbol{\phi}\|\leq\pi$. The equivalence above is stated on $SO(3)$ and is unaffected by this choice, but the search operators of Section~\ref{sec:proposed} parametrise boundary attitudes on the principal chart and therefore explore the principal homotopy class of rotations: candidate trajectories that reach the terminal attitude ``the long way around'' are admissible in $\mathcal{T}$ but are not generated by the implemented projection. All optimality statements about the implemented method are accordingly made with respect to this reachable subset of $\mathcal{T}$.

\section{PROPOSED APPROACH}
\label{sec:proposed}

The outer problem~\eqref{eq:outer} is non-convex, but it is not unstructured: its objective $L(\boldsymbol{\tau})$ is the optimal value of a convex program whose data depend smoothly on $\boldsymbol{\tau}$. We exploit this by descending the exact gradient of $L$, obtained from the envelope theorem, and by handling the multimodality of $L$ through random restarts. Decentralization then has a precise function: $N$ agents run $N$ independent restart streams in parallel and agree on the best result by max-consensus. Algorithm~\ref{alg:island} summarizes the loop each agent executes. Fig.~\ref{fig:methodology} summarizes the complete architecture, including the closed-loop recovery path of Section~\ref{sec:recovery}.

\begin{figure}
    \centering
    \includegraphics[width=0.95\linewidth]{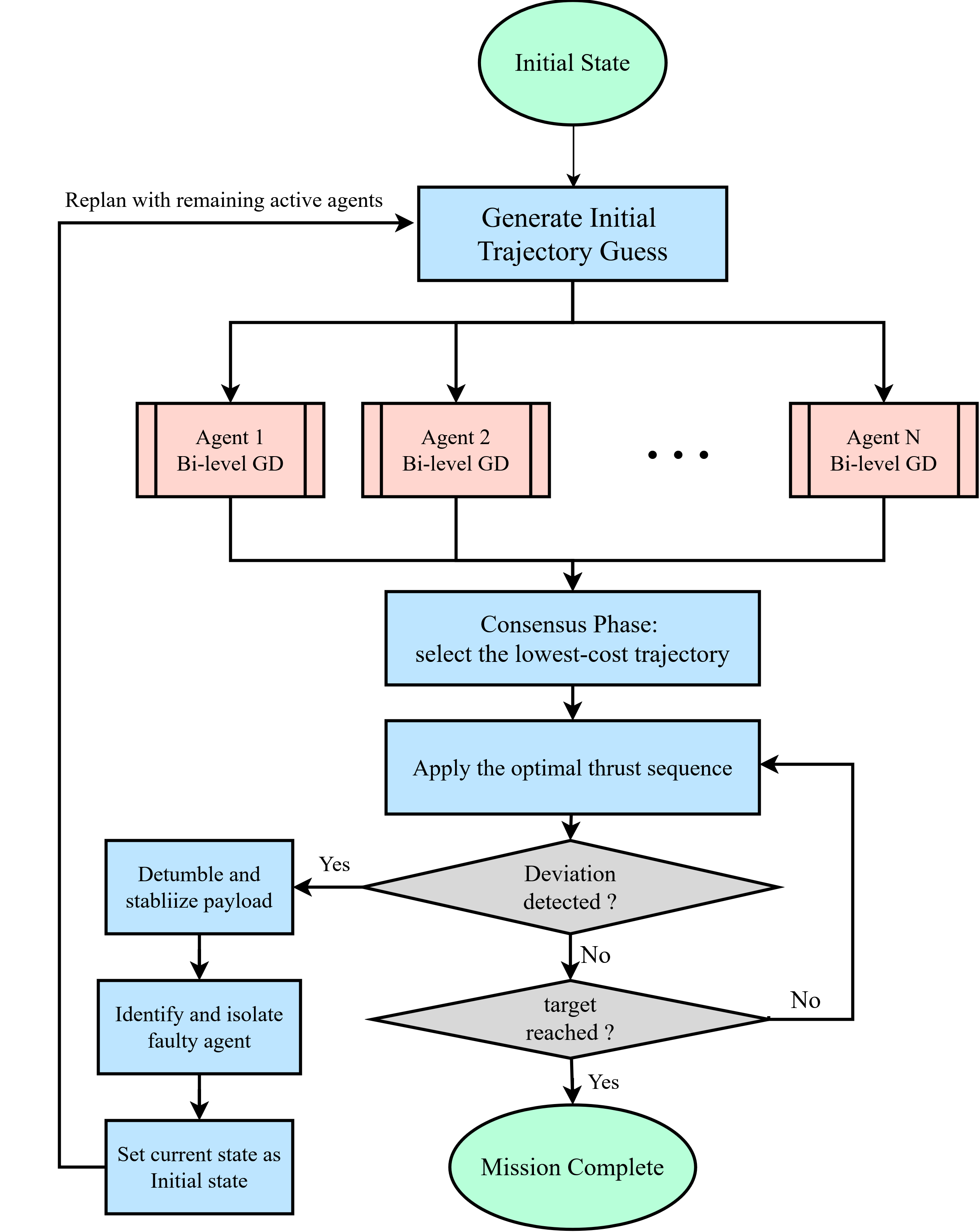}
    \caption{Methodology overview. From a shared initial trajectory guess, each agent runs an independent bi-level gradient-descent island (Sections~\ref{sec:envelope}--\ref{sec:multistart}); a max-consensus round selects the lowest-cost trajectory (Section~\ref{sec:consensus}), which the swarm executes. Detected deviations trigger the recovery loop of Section~\ref{sec:recovery}: detumble, isolate faulty agents, and replan with the survivors from the current state.}
    \label{fig:methodology}
\end{figure}

\subsection{Envelope-Theorem Gradient of the Inner Value}
\label{sec:envelope}

Write the inner problem~\eqref{eq:inner} in the generic form
\begin{equation}
  L(\boldsymbol{\tau})=\min_{\mathbf{U}}\; J(\mathbf{U})
  \quad\text{s.t.}\quad
  \boldsymbol{g}(\mathbf{U};\boldsymbol{\tau})=\mathbf{0},\;\;
  \boldsymbol{h}(\mathbf{U};\boldsymbol{\tau})\leq\mathbf{0},
  \label{eq:inner_generic}
\end{equation}
The torque trajectory enters the constraint functions through two channels: directly, in the torque-allocation equalities~\eqref{eq:inner_torque}, and indirectly, through the rotation matrices $R_k(\boldsymbol{\tau})$ obtained by propagating the attitude recursion~\eqref{eq:disc_att}--\eqref{eq:disc_omega}, which appear in the attachment vectors $\boldsymbol{q}_{i,k}=R_k(\boldsymbol{\tau})\boldsymbol{q}_i^{\mathcal{B}}$ of both the allocation equalities and the cone constraints. Let $\mathcal{L}(\mathbf{U},\boldsymbol{\lambda},\boldsymbol{\mu};\boldsymbol{\tau})=J(\mathbf{U})+\boldsymbol{\lambda}^{\top}\boldsymbol{g}+\boldsymbol{\mu}^{\top}\boldsymbol{h}$ denote the Lagrangian, and let $(\mathbf{U}^\star,\boldsymbol{\lambda}^\star,\boldsymbol{\mu}^\star)$ be a primal-dual solution of~\eqref{eq:inner_generic}. Under standard regularity conditions (LICQ and strict complementarity at $\mathbf{U}^\star$), the value function is differentiable and the envelope theorem~\cite{danskin_theory_1967} gives
\begin{equation}
  \nabla_{\boldsymbol{\tau}} L(\boldsymbol{\tau})
  = \nabla_{\boldsymbol{\tau}}\,
    \mathcal{L}(\mathbf{U}^\star,\boldsymbol{\lambda}^\star,\boldsymbol{\mu}^\star;\boldsymbol{\tau}).
  \label{eq:envelope}
\end{equation}
The minimizer's sensitivity $\partial\mathbf{U}^\star/\partial\boldsymbol{\tau}$ is not required: only the partial derivative of the Lagrangian at the converged point. In particular, no implicit differentiation through the solver's KKT system is performed. Since $\boldsymbol{\tau}\in\mathbb{R}^{3T}$, the gradient~\eqref{eq:envelope} has dimension $3T$: its size, and the cost of evaluating it, scale with the number of time steps and are independent of the number of agents.

In implementation, the inner problem is solved exactly as a second-order cone program (SOCP): the fuel objective enters in epigraph form, $\|\boldsymbol{u}_{i,k}\|\leq t_{i,k}$ with linear objective $\sum_{i,k}t_{i,k}$, the thrust cones~\eqref{eq:inner_cone} are exact second-order cone constraints, the translational dynamics are linear, and the rotation matrices $R_k(\boldsymbol{\tau})$ enter as numeric data obtained by rolling out the attitude recursion at the current $\boldsymbol{\tau}$. A small quadratic regulariser $\epsilon_r\sum_{i,k}\|\boldsymbol{u}_{i,k}\|^2$ with $\epsilon_r=10^{-3}$ is added because the pure fuel objective can leave the allocation non-unique, as several allocations of identical total fuel can realise the same torque, which destabilises the duals; the regulariser selects a unique minimiser. The SOCP is modelled in CVXPY~\cite{diamond_cvxpy_2016} and solved to certified optimality by the interior-point conic solver Clarabel~\cite{goulart_clarabel_2024}.

The gradient~\eqref{eq:envelope} is assembled from the two channels identified above. The direct channel, through the torque-allocation equalities~\eqref{eq:inner_torque}, is read off the solver at no cost, the partial derivative of the Lagrangian with respect to $\boldsymbol{\tau}_k$ is the negative of the corresponding equality multipliers. The indirect channel, through the rotation matrices $R_k(\boldsymbol{\tau})$ appearing in the attachment vectors, is evaluated by a single backward adjoint pass over the attitude recursion~\eqref{eq:disc_att}--\eqref{eq:disc_omega} using the right Jacobian of the $SO(3)$ exponential map. The pass has length $T$ and cost independent of $N$. The attitude adjoint was verified against algorithmic differentiation on a smooth surrogate value function to a relative accuracy of $10^{-7}$, and the assembled gradient against central finite differences of the SOCP value to a relative directional accuracy of $10^{-3}$, the noise floor set by the conic solver's tolerance.

\subsection{Feasibility Projection}
\label{sec:proj}

Candidate torque trajectories are not guaranteed to satisfy the rotational boundary conditions \eqref{eq:outer_att}--\eqref{eq:outer_omega}. Each candidate $\boldsymbol{\tau}_{0:T-1}$ is projected onto the feasible set by solving
\begin{subequations}
\label{eq:proj}
\begin{align}
  \min_{\hat{\boldsymbol{\tau}}_{0:T-1}}&\;
  \sum_{k=0}^{T-1}\|\hat{\boldsymbol{\tau}}_k-\boldsymbol{\tau}_k\|^2
  \quad\text{s.t.}\quad\eqref{eq:outer_dyn_att},\;\eqref{eq:outer_dyn_omega},
  \label{eq:proj_obj}\\
  &\;\boldsymbol{\phi}_0=\boldsymbol{\phi}_a,\;
   \boldsymbol{\phi}_T=\boldsymbol{\phi}_b,\\
  &\;\boldsymbol{\omega}_0=\boldsymbol{\omega}_a,\;
   \boldsymbol{\omega}_T=\boldsymbol{\omega}_b.
\end{align}
\end{subequations}
Although~\eqref{eq:proj} is non-convex, it is small-dimensional and solved quickly via a multiple-shooting NLP warm-started from $\boldsymbol{\tau}_{0:T-1}$, ensuring every candidate evaluated by the outer search is rotationally feasible. When the projection NLP fails to converge, which occurs with increasing frequency as the cone half-angle $\nu$ shrinks, the candidate is rejected rather than evaluated unprojected; the consequences of this for the reachable search space are made explicit in Section~\ref{sec:convergence}.

\subsection{Projected Gradient Descent}
\label{sec:pgd}

Because the outer feasible set $\mathcal{T}$ (rotational boundary conditions) is enforced by the projection~\eqref{eq:proj} rather than analytically, the descent is a projected gradient method on the composition $L\circ\Pi$, where $\Pi$ denotes the projection; the projection's derivative is neglected, as is standard in projected gradient schemes. Each iteration takes a trial step against the normalised gradient and re-projects:
\begin{equation}
  \boldsymbol{\tau}^{+}=\Pi\!\left(\boldsymbol{\tau}
  -\alpha\,\frac{\nabla_{\boldsymbol{\tau}}L(\boldsymbol{\tau})}{\|\nabla_{\boldsymbol{\tau}}L(\boldsymbol{\tau})\|}\right).
  \label{eq:pgd_step}
\end{equation}
The step length is managed by a backtracking line search on the true objective (each trial requires one projection and one inner solve): the initial step is chosen \emph{relative} to the torque scale, $\alpha_0=\eta\,\mathrm{RMS}(\boldsymbol{\tau})\sqrt{3T}$ with $\eta=0.03$, so that the first move displaces the trajectory by approximately $3\%$ of its root-mean-square magnitude; successful steps grow $\alpha$ by a factor $1.5$, failed steps shrink it by $0.5$, and the descent declares first-order stationarity when no improving step above a minimum relative length exists. The relative scaling is essential for transfer across scenarios whose torque magnitudes span several orders of magnitude.

\subsection{Random Restarts and Decentralized Multistart}
\label{sec:multistart}

The value function $L$ is multimodal in general; empirically the number of separated basins grows as the cone half-angle $\nu$ tightens. A single descent therefore only certifies a local minimum. Upon reaching stationarity, the method restarts from a random cold torque trajectory,
\begin{equation}
  \boldsymbol{\tau}^{(j)}_{0}\sim\mathcal{N}\!\left(\mathbf{0},\,
  s^2\mathbf{I}_{3T}\right),\qquad
  s=c\,\mathrm{RMS}(\boldsymbol{\tau}_{\mathrm{best}}),
  \label{eq:restart}
\end{equation}
with $c=1$, and retains the best pair $(\boldsymbol{\tau},L)$ found across all restarts. The very first descent is instead warm-started from a multiple-shooting nominal, which reliably lands in a good basin; note that re-solving the shooting problem does \emph{not} constitute a restart, as its outcome is deterministic, so genuine exploration requires the cold draws of~\eqref{eq:restart}.

In the decentralized variant, each agent $i\in\{1,\ldots,N\}$ is an independent \emph{island} executing the same descent-with-restarts loop within its own wall-clock budget. Island $1$ warm-starts from the shooting nominal; islands $2,\ldots,N$ initialise from independent random draws~\eqref{eq:restart} with per-island seeds, and all islands continue with their own random restarts as budget allows. No communication occurs during the search. The swarm thus evaluates $N$ restart streams per unit of per-agent wall-clock.

\begin{algorithm}[t]
\caption{Per-Agent Bi-Level Descent (island $i$)}
\label{alg:island}
\begin{algorithmic}[1]
\STATE \textbf{Input:} boundary conditions $(\boldsymbol{a},\boldsymbol{b})$; seed $s_i$; per-agent budget $t_{\max}$
\STATE $(\boldsymbol{\tau}^{\star}_{i},\,L^{\star}_{i}) \gets (\varnothing,\,+\infty)$
\STATE $\boldsymbol{\tau}_{0} \gets \texttt{shooting\_nominal()}$ \textbf{if} $i=1$ \textbf{else} draw~\eqref{eq:restart} (Sec.~\ref{sec:multistart})
\STATE \textbf{while} $\texttt{wall\_clock()} < t_{\max}$:
\STATE \hspace{1em} $\boldsymbol{\tau} \gets \Pi(\boldsymbol{\tau}_{0})$ by~\eqref{eq:proj} (Sec.~\ref{sec:proj})
\STATE \hspace{1em} \textbf{if} $\boldsymbol{\tau} = \varnothing$: $\boldsymbol{\tau}_{0} \gets$ draw~\eqref{eq:restart}; \textbf{continue}
\STATE \hspace{1em} $(\mathbf{U}^{\star},\boldsymbol{\lambda}^{\star},\boldsymbol{\mu}^{\star},\,L) \gets \texttt{solve\_inner}(\boldsymbol{\tau})$ by~\eqref{eq:inner} (Sec.~\ref{sec:envelope})
\STATE \hspace{1em} \textbf{while not} stationary \textbf{and} $\texttt{wall\_clock()} < t_{\max}$:
\STATE \hspace{2em} $\boldsymbol{g} \gets \nabla_{\boldsymbol{\tau}}L$ at $(\mathbf{U}^{\star},\boldsymbol{\lambda}^{\star},\boldsymbol{\mu}^{\star})$ by~\eqref{eq:envelope} (Sec.~\ref{sec:envelope})
\STATE \hspace{2em} $\alpha \gets \texttt{line\_search}(\boldsymbol{\tau},\boldsymbol{g})$ (Sec.~\ref{sec:pgd})
\STATE \hspace{2em} $\boldsymbol{\tau} \gets \Pi\bigl(\boldsymbol{\tau}-\alpha\,\boldsymbol{g}/\lVert\boldsymbol{g}\rVert\bigr)$ by~\eqref{eq:pgd_step}
\STATE \hspace{2em} $(\mathbf{U}^{\star},\boldsymbol{\lambda}^{\star},\boldsymbol{\mu}^{\star},\,L) \gets \texttt{solve\_inner}(\boldsymbol{\tau})$
\STATE \hspace{1em} \textbf{if} $L < L^{\star}_{i}$: $(\boldsymbol{\tau}^{\star}_{i},\,L^{\star}_{i}) \gets (\boldsymbol{\tau},\,L)$
\STATE \hspace{1em} $\boldsymbol{\tau}_{0} \gets$ draw~\eqref{eq:restart}
\STATE \textbf{return} $(\boldsymbol{\tau}^{\star}_{i},\,L^{\star}_{i})$ to max-consensus (Sec.~\ref{sec:consensus})
\end{algorithmic}
\end{algorithm}

\subsection{Convergence of the Random Multistart}
\label{sec:convergence}

Let $D$ denote the restart distribution~\eqref{eq:restart} on $\mathbb{R}^{3T}$ and let $A:\mathbb{R}^{3T}\to\mathcal{T}\cup\{\varnothing\}$ be the deterministic map taking a draw $\boldsymbol{\tau}_0$ to the stationary point returned by the projected descent of Section~\ref{sec:pgd}, with $A(\boldsymbol{\tau}_0)=\varnothing$ when the projection fails. Define the induced cost $\mathcal{C}(\boldsymbol{\tau}_0)=L(A(\boldsymbol{\tau}_0))$ with $\mathcal{C}=+\infty$ on failures, and the essential infimum
\begin{equation}
  L^{\ddagger}\;\triangleq\;\inf\bigl\{\,y\in\mathbb{R}\;:\;
  \mathbb{P}_{D}\bigl(\mathcal{C}(\boldsymbol{\tau}_0)\leq y\bigr)>0 \,\bigr\}.
\end{equation}

\begin{theorem}[Almost-sure convergence]
\label{thm:multistart}
Let $\boldsymbol{\tau}^{(1)}_0,\boldsymbol{\tau}^{(2)}_0,\ldots$ be i.i.d.\ draws from $D$ and let $M_n=\min\limits_{j\leq n}\mathcal{C}(\boldsymbol{\tau}^{(j)}_0)$. Then $M_n\to L^{\ddagger}$ almost surely as $n\to\infty$.
\end{theorem}

\begin{proof}
By definition of $L^{\ddagger}$, $\mathbb{P}(\mathcal{C}<L^{\ddagger})=0$, so almost surely every $\mathcal{C}(\boldsymbol{\tau}^{(j)}_0)\geq L^{\ddagger}$ and hence $M_n\geq L^{\ddagger}$ for all $n$. For any $\epsilon>0$, $p_\epsilon\triangleq\mathbb{P}(\mathcal{C}\leq L^{\ddagger}+\epsilon)>0$, and the events $\{\mathcal{C}(\boldsymbol{\tau}^{(j)}_0)\leq L^{\ddagger}+\epsilon\}$ are independent with common probability $p_\epsilon$; by the second Borel--Cantelli lemma one occurs almost surely, so $\limsup_n M_n\leq L^{\ddagger}+\epsilon$ almost surely. Since $M_n$ is non-increasing and $\epsilon$ arbitrary, $M_n\to L^{\ddagger}$ almost surely.
\end{proof}

\textit{Remark 2 (what the limit is, and is not).} $L^{\ddagger}$ is the global optimum of $V$ over the \emph{reachable set} of the composed operator, that is, over stationary points attainable by projected descent from initialisations that the projection accepts. If the basin of attraction of the global minimiser of $V$ over the principal chart has positive $D$-measure and the projection succeeds with positive probability inside it, then $L^{\ddagger}$ coincides with that global optimum. Empirically the projection's acceptance probability from cold draws decays as $\nu$ shrinks, and below $\nu\approx2.5^\circ$ it can vanish for some scenarios, shrinking the reachable set. The method converges to the best solution its projection can reach, not unconditionally to the global optimum of~\eqref{eq:outer}.

\textit{Remark 3 (decentralized speed-up).} With $N$ islands drawing independently from $D$, the first $n$ restarts complete in $\lceil n/N\rceil$ rounds of per-agent wall-clock, so Theorem~\ref{thm:multistart} holds along the swarm's aggregate sequence with an $N$-fold reduction in elapsed time. The warm-started island additionally guarantees the swarm never does worse than the deterministic single-start baseline.

\subsection{Consensus Round}
\label{sec:consensus}

After the search budget elapses, the swarm agrees on the best solution by a single max-consensus round on the pair $(1/L_i,\,i)$, where $L_i$ is island $i$'s best cost: at each synchronous iteration every agent replaces its pair by the maximum (in the first coordinate) over itself and its neighbours in the line-of-sight graph $\mathcal{G}$. Islands whose every start failed carry fitness $0$ and cannot win. Provided $\mathcal{G}$ is connected, all agents hold the global winner after at most $\mathrm{diam}(\mathcal{G})$ iterations, the winning agent broadcasts its torque trajectory along the same tree, and each agent recovers its own thrust sequence from one final inner solve of~\eqref{eq:inner}. Connectedness of $\mathcal{G}$ is assumed throughout.

\subsection{Closed-Loop Fault Recovery}
\label{sec:recovery}

The decentralized planner is embedded in a closed-loop architecture for in-flight faults. Each agent tracks the planned thrust sequence while monitoring two triggers: a trajectory-deviation monitor and a per-agent communication heartbeat with a timeout, under heterogeneous communication delays. When a trigger fires, the swarm executes a three-phase recovery: (i) \emph{detumble}, a cone-projected proportional-derivative braking phase that brings the payload to rest so that replanning starts from a well-defined boundary condition; (ii) \emph{identify}, in which unresponsive or misbehaving agents are flagged and removed from the active set; and (iii) \emph{replan}, in which the surviving agents run the decentralized multistart of Section~\ref{sec:multistart} from the current state to the original terminal condition, followed by the consensus round. If the survivors' attachment geometry no longer spans the required wrench space, the replanning problem is infeasible. The planner then reports the failure explicitly rather than returning an unattainable plan, and the episode is scored as a survivability failure. 

\section{NUMERICAL EXPERIMENTS}
\label{sec:experiments}

\subsection{Baseline Methods}
\label{sec:baselines}

The proposed method is compared against two sampling-based outer searches and a direct joint NLP. All baselines share the bi-level structure of Section~\ref{sec:reform}: the same inner solver, the same feasibility projection~\eqref{eq:proj}, and the same multiple-shooting warm start, so the comparison isolates the outer search strategy.

\subsubsection{Genetic Algorithm (GA)}
A population of $P$ flattened torque trajectories $\boldsymbol{\tau}\in\mathbb{R}^{3T}$ is evolved with two-point crossover, fixed-percentage mutation, and steady-state selection with elitism. Every candidate is projected via~\eqref{eq:proj} before evaluation, and fitness is the reciprocal inner cost. The GA is implemented with PyGAD~\cite{gad_pygad_2021}.

\subsubsection{Model Predictive Path Integral (C-MPPI / D-MPPI)}
Path-integral search~\cite{williams_information_2017} perturbs a nominal trajectory with $S$ Gaussian samples $\boldsymbol{\varepsilon}^{(s)}\sim\mathcal{N}(\mathbf{0},\sigma^2\mathbf{I})$, projects each candidate, $\hat{\boldsymbol{\tau}}^{(s)}=\Pi(\boldsymbol{\tau}^{(t)}+\boldsymbol{\varepsilon}^{(s)})$, scores it with the inner solver, and updates the nominal toward the weighted mean of the \emph{projected} candidates:
\begin{align}
  \boldsymbol{\tau}^{(t+1)} &= \boldsymbol{\tau}^{(t)}
  + \gamma\!\left(\sum_{s=1}^{S}\frac{w^{(s)}}{\sum_j w^{(j)}}\,\hat{\boldsymbol{\tau}}^{(s)}
  - \boldsymbol{\tau}^{(t)}\right),
  \label{eq:mppi_update}\\
  w^{(s)} &= \exp\!\left(-\frac{J^{(s)}-J_{\min}}{\lambda}\right),
  \notag
\end{align}
with temperature $\lambda>0$ and step size $\gamma\in(0,1]$. Weighting the projected candidates, rather than the raw perturbations, keeps the update target inside the set of points that were actually evaluated; with raw perturbations the scores and the update are computed at different points, which degrades convergence. The perturbation scale is set \emph{relative} to the nominal's RMS magnitude, mirroring Section~\ref{sec:pgd}. Failed inner solves carry cost $+\infty$ and zero weight, and the best-ever projected candidate is retained. The decentralised variant runs one MPPI island per agent from independent seeds, followed by the consensus round.

\subsubsection{Direct NLP}
The original problem~\eqref{eq:ocp} is also solved directly as a single joint NLP over $(\mathbf{U},\boldsymbol{x}_{0:T})$ with the same Tschauner--Hempel translational dynamics, body-frame thrusts, and Lie-group attitude kinematics as the bi-level pipeline, cold-started from the multiple-shooting nominal.

\subsection{Experimental Design}
\label{sec:design}

Scenarios are randomised by sampling orbital parameters, payload properties, terminal boundary conditions, agent count, and attachment vectors from uniform distributions over mission-representative ranges (semi-major axis $1.1$--$1.3\,R_\oplus$, eccentricity up to $0.30$, payload mass and principal inertia moments up to $5000$ in SI units, $3$--$10$ agents, and terminal states sampled uniformly over position and $SO(3)$). The payload begins at rest with zero rotation and must reach the sampled terminal position and attitude at rest. Two metrics are recorded per solver call: total fuel $J = \sum_{i,k}\|\boldsymbol{u}_{i,k}\|$ and terminal state violation
\begin{equation}
  \mathbf{V} = \|\boldsymbol{r}_T - \boldsymbol{r}_b\|
             + \|\boldsymbol{v}_T - \boldsymbol{v}_b\|
             + \bigl\|\log(R_T^\top R_b)\bigr\|
             + \|\boldsymbol{\omega}_T - \boldsymbol{\omega}_b\|,
  \label{eq:violation}
\end{equation}
where the attitude term is the geodesic distance on $SO(3)$. Every returned control sequence is re-simulated in a separate integrator, and the reported fuel and violation are those of the rollout, never the solver's own claim. A run is scored successful when the rollout satisfies $\mathbf{V}<10^{-3}$. Decentralised methods follow the per-agent budget convention. Each island receives the full wall-clock budget, reflecting hardware in which every agent is an independent computer, and islands are executed serially in simulation.

The baseline grid evaluates six methods (centralised and decentralised GD and MPPI, centralised GA, and the NLP) on $500$ randomised trials at every combination of cone half-angle $\nu\in\{1^\circ, 2^\circ, 2.5^\circ, 5^\circ, 15^\circ, 60^\circ\}$ and per-agent compute budget $\{10, 30, 60, 300, 600\}\,\mathrm{s}$, enforced as a hard CPU-time cap on every solver call; after removing duplicated scenario draws, all statistics are computed over the $472$ unique scenarios per angle. The same data support the decentralisation analysis, in which each decentralised method is compared with its centralised counterpart pairwise per scenario. The scalability study runs the closed-loop architecture of Section~\ref{sec:recovery} at fleet sizes $N\in\{4,11,18,25\}$ in two conditions: open loop (no faults injected) and closed loop (randomised agent-dropout faults), recording cost, terminal violation, recovery cycles, removed agents, and the replan-feasibility status that defines the fault-survivability statistics.

\section{RESULTS AND DISCUSSION}
\label{sec:results}

\begin{figure}[!t]
\centering
\includegraphics[width=0.92\columnwidth]{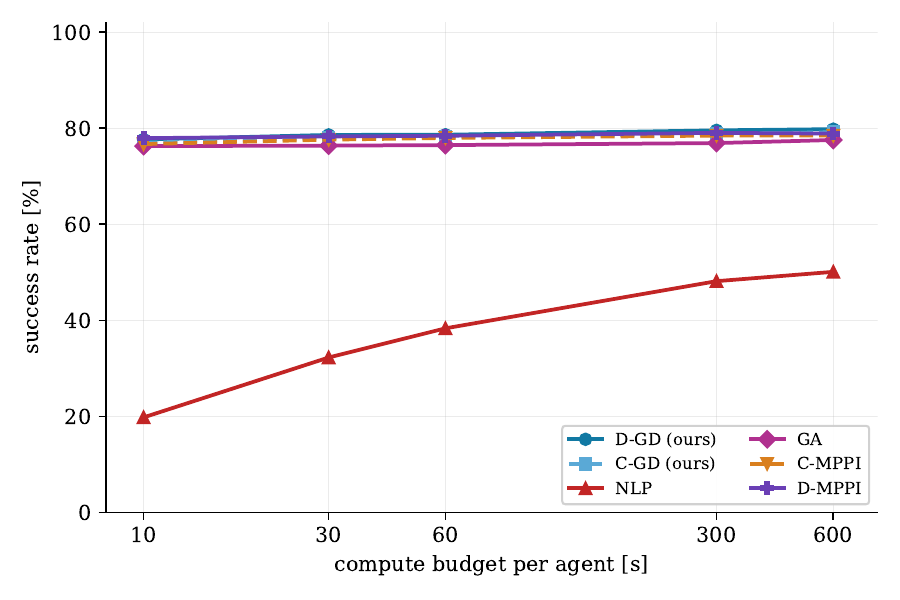}\\[2pt]
\includegraphics[width=0.92\columnwidth]{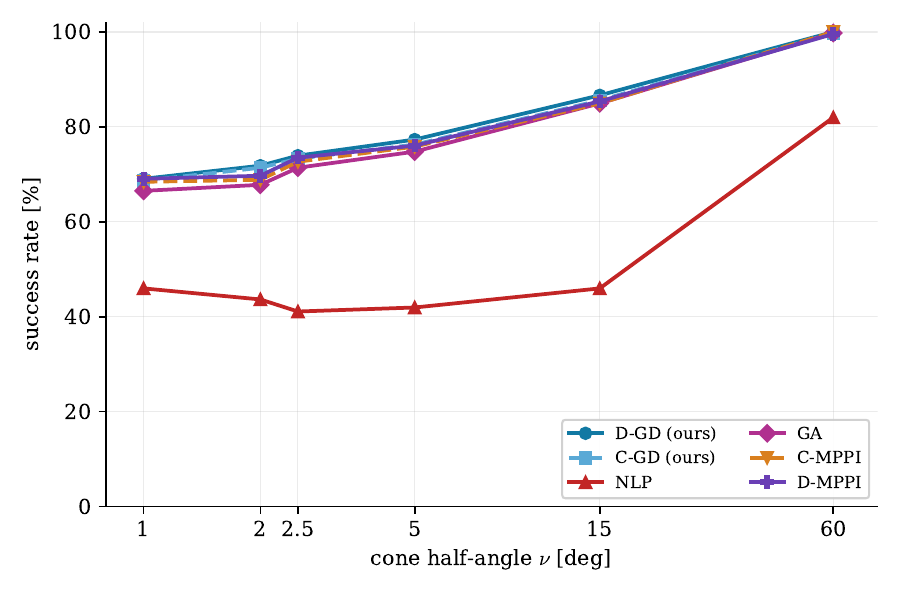}
\caption{Success rate, where a run succeeds when the independent verification rollout of the returned controls satisfies $\mathbf{V}<10^{-3}$. Top: versus per-agent compute budget, pooled over all cone half-angles ($2832$ scenarios per point). Bottom: versus cone half-angle at the $600\,\mathrm{s}$ budget ($472$ scenarios per point).}
\label{fig:success}
\end{figure}

\begin{figure}[!t]
\centering
\includegraphics[width=0.92\columnwidth]{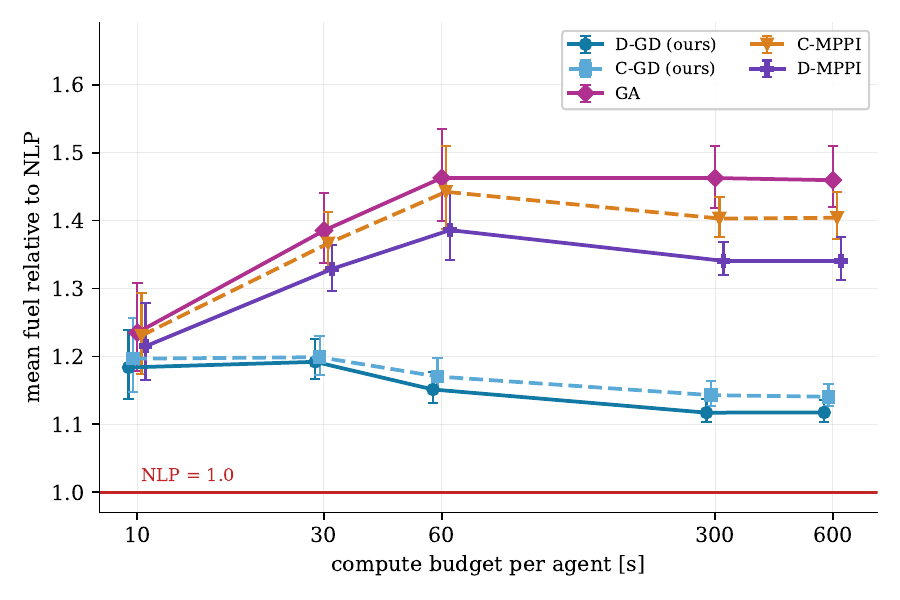}\\[2pt]
\includegraphics[width=0.92\columnwidth]{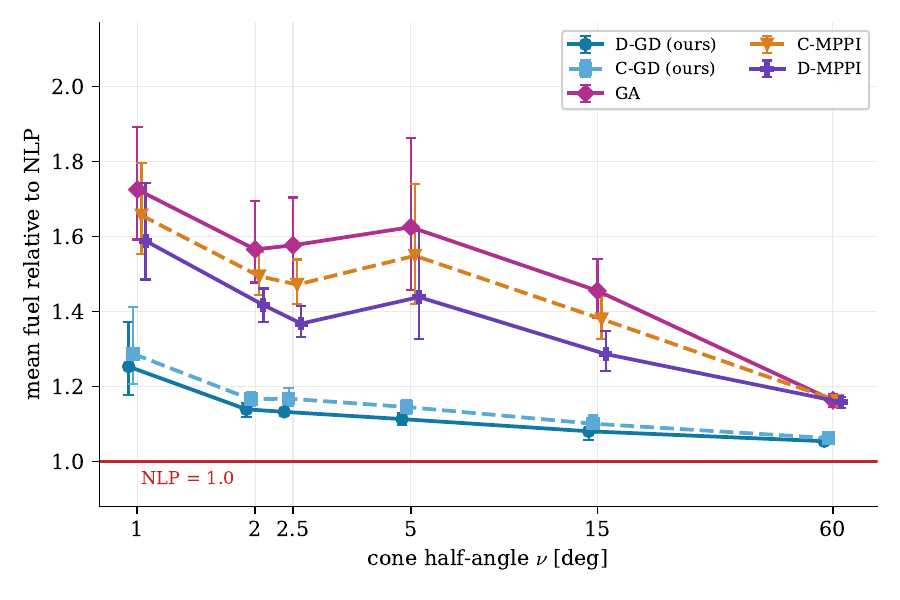}
\caption{Mean fuel relative to the NLP on jointly solved scenarios; each per-scenario ratio is paired before averaging, and error bars are bootstrap $95\%$ confidence intervals of the mean paired ratio. Top: versus per-agent compute budget, pooled over cone half-angles. Bottom: versus cone half-angle at the $600\,\mathrm{s}$ budget.}
\label{fig:fuel}
\end{figure}

\begin{table}[!t]
\centering
\caption{Success Rates of NLP and Proposed Method [NLP \%\,/\,D-GD \%].}
\label{tab:success_fact}
{\small
\setlength{\tabcolsep}{4pt}
\begin{tabular}{@{}lccccc@{}}
\toprule
$\nu$ & $10$\,s & $30$\,s & $60$\,s & $300$\,s & $600$\,s \\
\midrule
$1^\circ$   & 10.0\,/\,66.3 & 22.9\,/\,68.0 & 32.0\,/\,68.0 & 43.4\,/\,69.5 & 46.0\,/\,69.1 \\
$2^\circ$   & 12.1\,/\,67.8 & 22.7\,/\,69.7 & 29.4\,/\,70.1 & 40.7\,/\,71.4 & 43.6\,/\,71.8 \\
$2.5^\circ$ & 13.1\,/\,72.5 & 22.7\,/\,73.1 & 27.8\,/\,73.1 & 40.5\,/\,73.9 & 41.1\,/\,73.9 \\
$5^\circ$   & 12.5\,/\,73.9 & 23.7\,/\,75.0 & 28.8\,/\,75.0 & 39.4\,/\,76.1 & 41.9\,/\,77.3 \\
$15^\circ$  & 19.5\,/\,85.4 & 25.8\,/\,85.8 & 33.7\,/\,85.6 & 44.1\,/\,86.4 & 46.0\,/\,86.7 \\
$60^\circ$  & 51.7\,/\,100.0 & 75.8\,/\,100.0 & 78.4\,/\,100.0 & 80.9\,/\,100.0 & 82.0\,/\,100.0 \\
\bottomrule
\end{tabular}}
\end{table}

\begin{figure}[!t]
\centering
\includegraphics[width=0.92\columnwidth]{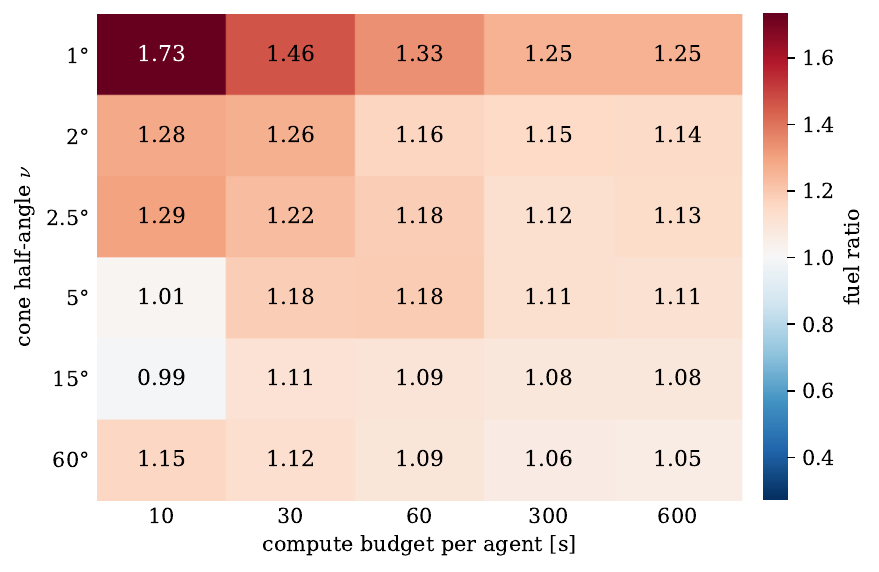}
\caption{Mean paired fuel ratio of D-GD relative to the NLP over the full angle--budget factorial, on jointly solved scenarios. Cells near or below one at $10\,\mathrm{s}$ occur where the NLP is truncated at the CPU cap, so its returned iterate is feasible but far from fuel-optimal.}
\label{fig:fuel_heatmap}
\end{figure}

\subsection{Baseline Comparison}
\label{sec:results_baseline}

Fig.~\ref{fig:success} reports reliability. The gradient-descent methods and the sampling baselines occupy a band far above the joint NLP at every budget. Pooled over angles, D-GD succeeds on $77.6$--$79.8\%$ of scenarios as the per-agent budget grows from $10$ to $600\,\mathrm{s}$, C-GD on $77.0$--$79.2\%$, and the samplers sit flat at $76$--$79\%$, while the NLP starts at $19.8\%$ and climbs to $50.1\%$. The gap closes as compute lengthens but remains large: at $10\,\mathrm{s}$ the NLP solves one quarter as many scenarios as D-GD, and at $600\,\mathrm{s}$ it still trails by $30$ percentage points. The angle sweep at $600\,\mathrm{s}$ shows the same separation at every cone width. The NLP reaches $41$--$46\%$ for $\nu\leq15^\circ$ and $82\%$ even at the benign $\nu=60^\circ$, against $69$--$87\%$ and $100\%$ for D-GD, which matches or exceeds every baseline in every angle cell and posts the highest success rate at every budget of $30\,\mathrm{s}$ and above (Table~\ref{tab:success_fact}). The NLP's dominant failure mode at short budgets is truncation, the CPU cap interrupts the joint solve mid-descent and the returned iterate, re-simulated, misses the terminal set. The bi-level methods degrade gracefully under the same cap because the inner allocation is solved exactly at every outer iterate, so an interrupted outer search still returns dynamically consistent, cone-feasible controls.

Fuel tells the complementary story on the scenarios the NLP does solve (Fig.~\ref{fig:fuel}). Paired per scenario against the NLP, the gradient methods are the closest of the outer searches at every budget. The mean D-GD ratio falls from $1.18$ at $10\,\mathrm{s}$ to $1.12$ at $600\,\mathrm{s}$, and C-GD from $1.20$ to $1.14$, whereas the samplers drift upward with budget, from $1.21$--$1.24$ at $10\,\mathrm{s}$ to $1.34$--$1.46$ at $600\,\mathrm{s}$. The takeaway is, added compute buys the samplers more successes but not better optima, while it buys the gradient methods both. The direct comparison between the NLP and the proposed method (Fig.~\ref{fig:fuel_heatmap}) illustrates the tradeoff more clearly. At tight cones under starved budgets, the proposed method costs significantly more fuel than the NLP baseline (up to $1.73$ at $\nu=1^\circ$, $10\,\mathrm{s}$), though it is several times more reliable. Both the reliability advantage of the proposed approach, and the fuel-use advantage of the NLP baseline, decay along both axes, with D-GD within $5$--$8\%$ of the converged NLP cost for $\nu\geq15^\circ$ at all budgets of $300\,\mathrm{s}$ and above. The near-parity cells at $10\,\mathrm{s}$ ($0.99$--$1.01$ at $\nu=5^\circ$--$15^\circ$) are the truncation effect in reverse. There the NLP's surviving iterates are feasible but expensive, and D-GD matches them outright.

\subsection{Value of Decentralisation}
\label{sec:results_decentralisation}

The decentralised multistart of Section~\ref{sec:multistart} gives a consistent, modest edge over the centralised variant. D-GD's mean paired fuel ratio matches or betters C-GD's at every budget, $1.18$ against $1.20$ at $10\,\mathrm{s}$ and $1.12$ against $1.14$ at $600\,\mathrm{s}$: the independent islands explore at least as many basins per second of wall-clock as a single centralised restart chain, and the two variants stay within $0.02$ of one another once either has explored the reachable basins. D-GD also posts the higher success rate at every budget (Fig.~\ref{fig:success}). Decentralisation therefore costs nothing in solution quality and is what enables the closed-loop recovery mode of Section~\ref{sec:recovery}.

\subsection{Scalability and Fault Survivability}
\label{sec:results_scalability}

\begin{figure}[!t]
\centering
\includegraphics[width=0.92\columnwidth]{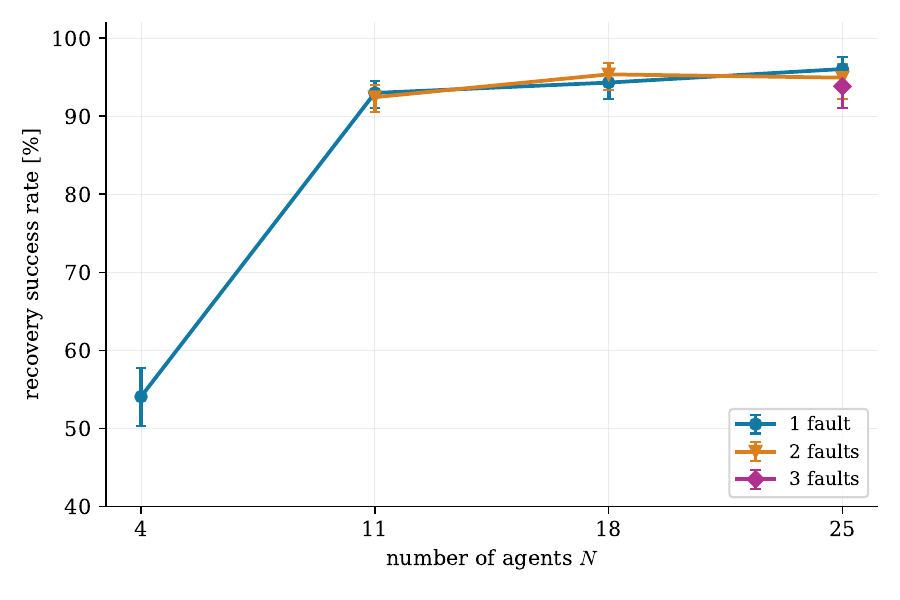}
\caption{Closed-loop fault recovery success versus fleet size under injected actuation faults ($700$--$1800$ episodes per fleet size; error bars are Wilson $95\%$ confidence intervals). For $N\geq11$ recovery is insensitive to whether one, two, or three agents fail simultaneously.}
\label{fig:recovery}
\end{figure}

Fig.~\ref{fig:recovery} reports the closed-loop architecture under randomised actuation faults that disable one to three agents, with the survivors detecting the fault, re-deriving boundary conditions in flight, and replanning decentrally. Recovery reliability improves with fleet size: $54\%$ of episodes recover at $N=4$, rising to $93\%$ at $N=11$ and $94$--$96\%$ at $N\in\{18,25\}$. For $N\geq11$ the rate is insensitive to the number of simultaneous faults, with two-fault episodes recovering at $92$--$95\%$ and three-fault episodes at $94\%$ where sampled ($N=25$). The four-agent floor is physical rather than algorithmic since losing one of four agents removes a quarter of the fleet's actuation authority, and the dominant failure is a wrench-infeasible replan in which the survivors' attachment cones no longer span the wrench required to brake and re-steer the payload. Median replan wall-clock is constant at roughly $29\,\mathrm{s}$ per agent across all fleet sizes under serial island simulation, so recovery latency is flat on distributed hardware. Fault tolerance is thus an emergent benefit of redundancy that the decentralised planner converts into recovered missions at scale.

\section{CONCLUSION}
\label{sec:conclusion}

This paper presented a decentralised bi-level framework for cooperative orbital transportation of a rigid payload. The problem is split into an inner convex thrust allocation and an outer search over payload torque trajectories, proven equivalent to the original problem, and the outer search is solved by projected gradient descent on the exact envelope-theorem gradient of the inner optimal value. Random restarts handle multimodality, each agent runs an independent restart stream, and a max-consensus round selects the swarm-wide best trajectory; almost-sure convergence to the optimum of the projector-reachable set was established, and the same planner drives a closed-loop fault-recovery mode. 

In Monte Carlo experiments across cone half-angles from $1^\circ$ to $60^\circ$ and per-agent budgets from $10$ to $600\,\mathrm{s}$, the method matched or exceeded every baseline's reliability at every cone angle, succeeding on $78$--$80\%$ of rollout-verified scenarios against $20$--$50\%$ for the direct NLP. It came closest to the converged NLP's fuel among all outer searches, within $5$--$8\%$ at moderate cone angles and generous budgets, where the sampling baselines paid $34$--$46\%$, it was never worse than its centralised counterpart on fuel or reliability at any budget, and it raised closed-loop fault survivability from $54\%$ at four agents to $93$--$96\%$ at eleven or more. The proposed approach effectively takes problems that are infeasible for NLP solvers and returns feasible, but more expensive, solutions for them, and the gap between the proposed approach and NLP solvers reduces as the problem becomes easier and compute budget increases.

More broadly, the results indicate that the exact sensitivity of a convex subproblem can turn a heuristic decomposition into a principled first-order method, and that decentralisation is most valuable not as a way to split computation but as a way to hedge against the multimodality of the underlying trajectory landscape. Both observations extend to other multi-agent space tasks, such as on-orbit assembly and servicing, in which a shared rigid body couples the agents' actuation. Future work will widen the projection's reachable set at tight cone angles, through a wider projection funnel or a homotopy over the cone angle, adapt the restart distribution online toward unexplored basins, and extend the closed-loop layer to time-varying communication topologies and to faults during the detumble phase.  Together, these steps move toward swarms that can be handed a large or non-cooperative object in orbit and be trusted to move it, fuel-efficiently and without a single point of failure, under the faults and constraints of real missions.

\section{Acknowledgements}
The authors would like to acknowledge the assistance given by Research IT and the use of the Computational Shared Facility at The University of Manchester. Large language model tools (Anthropic Claude models, Claude Opus 4.8 and Claude Fable 5) assisted with code development and manuscript editing; all methods, results, analyses, and conclusions were developed and verified by the authors, and no generative artificial intelligence was used to create any figure or image.

\bibliographystyle{IEEEtaes}
\bibliography{references}

\begin{thebibliography}{10}
\providecommand{\url}[1]{#1}
\csname url@samestyle\endcsname
\renewcommand{\newblock}{\par}
\providecommand{\bibinfo}[2]{#2}
\providecommand{\BIBentrySTDinterwordspacing}{\spaceskip=0pt\relax}
\providecommand{\BIBentryALTinterwordstretchfactor}{4}
\providecommand{\BIBentryALTinterwordspacing}{\spaceskip=\fontdimen2\font plus
\BIBentryALTinterwordstretchfactor\fontdimen3\font minus
  \fontdimen4\font\relax}
\providecommand{\BIBforeignlanguage}[2]{{%
\expandafter\ifx\csname l@#1\endcsname\relax
\typeout{** WARNING: IEEEtran.bst: No hyphenation pattern has been}%
\typeout{** loaded for the language `#1'. Using the pattern for}%
\typeout{** the default language instead.}%
\else
\language=\csname l@#1\endcsname
\fi
#2}}
\providecommand{\BIBdecl}{\relax}
\BIBdecl

\bibitem{wang_safe_2025}
K.~Wang, J.~Qiao, E.~Wang, J.~Lei, W.~Wang, and T.~Meng
\newblock  Safe inspection tracking control for a tumbling target under full
  state constraints and limited thrust \newblock  \emph{IEEE Transactions on
  Aerospace and Electronic Systems}, vol.~61, no.~3, pp. 6822--6838, 2025.

\bibitem{farivarnejad_multirobot_2022}
H.~Farivarnejad and S.~Berman
\newblock  Multirobot {Control} {Strategies} for {Collective} {Transport}
  \newblock  \emph{Annual Review of Control, Robotics, and Autonomous Systems},
  vol.~5, no.~1, pp. 205--219, 2022.

\bibitem{leitner_multi-robot_2009}
J.~Leitner
\newblock  Multi-robot {Cooperation} in {Space}: {A} {Survey} \newblock  In
  \emph{2009 {Advanced} {Technologies} for {Enhanced} {Quality} of {Life}},
  2009, pp. 144--151.

\bibitem{stuber_lets_2020}
J.~Stüber, C.~Zito, and R.~Stolkin
\newblock  Let's {Push} {Things} {Forward}: {A} {Survey} on {Robot} {Pushing}
  \newblock  \emph{Frontiers in Robotics and AI}, vol.~7, p.~8, 2020.

\bibitem{abeshtan_robust_2024}
P.~Abeshtan and F.~Saghafi
\newblock  Robust trajectory tracking for cooperative aerial load
  transportation \newblock  \emph{IEEE Transactions on Aerospace and Electronic
  Systems}, vol.~60, no.~6, pp. 8448--8457, 2024.

\bibitem{zhang_decentralized_2020}
L.~Zhang, Y.~Sun, A.~Barth, and O.~Ma
\newblock  Decentralized control of multi-robot system in cooperative object
  transportation using deep reinforcement learning \newblock  \emph{IEEE
  Access}, vol.~8, pp. 184\,109--184\,119, 2020.

\bibitem{muhammed_realtime_2026}
I.~Muhammed, A.~A. Nada, and H.~El-Hussieny
\newblock  Real-time decentralized model predictive control for cooperative
  multi-robot object transport: Experimental validation \newblock
  \emph{Scientific Reports}, vol.~16, no.~1, p. 9824, 2026.

\bibitem{farivarnejad_fully_2021}
H.~Farivarnejad, A.~S. Lafmejani, and S.~Berman
\newblock  Fully {Decentralized} {Controller} for {Multi}-{Robot} {Collective}
  {Transport} in {Space} {Applications} \newblock  In \emph{2021 {IEEE}
  {Aerospace} {Conference} (50100)}, 2021, pp. 1--9.

\bibitem{chen_occlusion-based_2015}
J.~Chen, M.~Gauci, W.~Li, A.~Kolling, and R.~Groß
\newblock  Occlusion-{Based} {Cooperative} {Transport} with a {Swarm} of
  {Miniature} {Mobile} {Robots} \newblock  \emph{IEEE Transactions on
  Robotics}, vol.~31, no.~2, pp. 307--321, 2015.

\bibitem{sucan_advantages_2011}
I.~A. Şucan and L.~E. Kavraki
\newblock  On the advantages of task motion multigraphs for efficient mobile
  manipulation \newblock  In \emph{2011 {IEEE}/{RSJ} {International}
  {Conference} on {Intelligent} {Robots} and {Systems}}, 2011, pp. 4621--4626.

\bibitem{xu_cooperative_2025}
T.~Xu, W.~Wang, H.~Chai, and Q.~Xu
\newblock  Cooperative three-satellite methodology for space target detumbling
  \newblock  \emph{Advances in Space Research}, vol.~75, no.~10, pp.
  7290--7305, 2025.

\bibitem{meng_decentralized_2018}
Y.~Meng, Q.~Chen, and A.~Rahmani
\newblock  A decentralized cooperative control scheme for a distributed space
  transportation system \newblock  \emph{Robotics and Autonomous Systems}, vol.
  101, pp. 1--19, 2018.

\bibitem{vance_evaluation_2020}
L.~D. Vance, E.~Asphaug, and J.~Thangavelautham
\newblock  Evaluation of {Mother}-{Daughter} {Architectures} for {Asteroid}
  {Belt} {Exploration} \newblock  2020, arXiv:1812.11243 [astro-ph].

\bibitem{underwood_using_2015}
C.~Underwood, S.~Pellegrino, V.~J. Lappas, C.~P. Bridges, and J.~Baker
\newblock  Using {CubeSat}/micro-satellite technology to demonstrate the
  {Autonomous} {Assembly} of a {Reconfigurable} {Space} {Telescope} ({AAReST})
  \newblock  \emph{Acta Astronautica}, vol. 114, pp. 112--122, 2015.

\bibitem{udrea_cooperative_2015}
B.~Udrea and M.~Nayak
\newblock  A cooperative multi-satellite mission for controlled active debris
  removal from low {Earth} orbit \newblock  In \emph{2015 {IEEE} {Aerospace}
  {Conference}}, 2015, pp. 1--15.

\bibitem{marchesini_multiagent_2026}
G.~Marchesini \emph{et~al.}
\newblock  Multi-agent transportation of free-flyers in microgravity via
  pushing interaction under human-in-the-loop control \newblock  2026,
  arXiv:2609.24376.

\bibitem{asri_introductory_2024}
E.~G. Asri and Z.~H. Zhu
\newblock  An introductory review of swarm technology for spacecraft on-orbit
  servicing \newblock  \emph{International Journal of Mechanical System
  Dynamics}, vol.~4, no.~1, pp. 3--21, 2024.

\bibitem{adams_swarm_2026}
C.~A. Adams, T.~Smith, A.~Woodard, B.~C. Kempa, J.~D. Frank, and E.~Van~Kints
\newblock  Multi-agent swarm state of the art report \newblock  National
  Aeronautics and Space Administration, Tech. Rep. NASA/TP-20260005089, Jun.
  2026.

\bibitem{ma_adaptive_2023}
Y.~Ma, B.~Jiang, J.~Wang, and J.~Gong
\newblock  Adaptive fault-tolerant formation control for heterogeneous
  {UAVs-UGVs} systems with multiple actuator faults \newblock  \emph{IEEE
  Transactions on Aerospace and Electronic Systems}, vol.~59, no.~5, pp.
  6705--6716, 2023.

\bibitem{santos_reconfigurable_2025}
W.~G. Santos, P.~Mason, E.~T. Stoneking, and B.~V. Sarli
\newblock  Reconfigurable guidance strategy for compensating actuator faults in
  spacecraft formation flying \newblock  \emph{Journal of Guidance, Control,
  and Dynamics}, vol.~48, no.~2, pp. 282--296, 2025.

\bibitem{azadmanesh_decoupling_2026}
M.~Azadmanesh, K.~Georgiev, S.~Kolev, and M.~Todorov
\newblock  Decoupling global and local faults in satellite swarms using
  smart-freeze adaptation and isolation-priority logic \newblock
  \emph{Aerospace}, vol.~13, no.~2, p. 176, 2026.

\bibitem{chang_decentralized_2026}
C.~Chang and M.~Mesbahi
\newblock  Decentralized multiagent trajectory planning via duality and
  biconvex optimization \newblock  In \emph{AIAA SCITECH 2026 Forum}, 2026,
  aIAA 2026-0322.

\bibitem{morgan_model_2014}
D.~Morgan, S.-J. Chung, and F.~Y. Hadaegh
\newblock  Model predictive control of swarms of spacecraft using sequential
  convex programming \newblock  \emph{Journal of Guidance, Control, and
  Dynamics}, vol.~37, no.~6, pp. 1725--1740, 2014.

\bibitem{cui_trajectory_2024}
B.~Cui, X.~Chen, R.~Chai, Y.~Xia, and H.-S. Shin
\newblock  Trajectory planning of spacecraft swarm reconfiguration using
  reachable set-based collision constraints \newblock  \emph{IEEE Transactions
  on Aerospace and Electronic Systems}, vol.~60, no.~5, pp. 6474--6487, 2024.

\bibitem{liu_hierarchical_2026}
C.~Liu, Y.~Luo, X.~Yue, and E.~Capello
\newblock  A hierarchical optimization approach for spacecraft swarm
  trajectories considering dynamic communication topologies \newblock
  \emph{Aerospace Science and Technology}, vol. 168, p. 111250, 2026.

\bibitem{sikka_orbit_2025}
S.~D. Sikka, Z.~Lu, A.~Rai, D.~DeLaurentis, and S.~Mou
\newblock  On {Orbit} {Object} {Transportation} {With} {Spacecraft} {Swarms}
  \newblock  In \emph{{AIAA} {SCITECH} 2025 {Forum}}, 2025.

\bibitem{danskin_theory_1967}
J.~M. Danskin
\newblock \emph{The Theory of Max-Min and its Application to Weapons Allocation
  Problems}. Springer-Verlag, 1967.

\bibitem{diamond_cvxpy_2016}
S.~Diamond and S.~Boyd
\newblock  {CVXPY}: A {Python}-embedded modeling language for convex
  optimization \newblock  \emph{Journal of Machine Learning Research}, vol.~17,
  no.~83, pp. 1--5, 2016.

\bibitem{goulart_clarabel_2024}
P.~J. Goulart and Y.~Chen
\newblock  Clarabel: An interior-point solver for conic programs with quadratic
  objectives \newblock  2024, arXiv:2405.12762.

\bibitem{gad_pygad_2021}
A.~F. Gad
\newblock  {PyGAD}: {An} {Intuitive} {Genetic} {Algorithm} {Python} {Library}
  \newblock  2021, arXiv:2106.06158 [cs].

\bibitem{williams_information_2017}
G.~Williams \emph{et~al.}
\newblock  Information theoretic {MPC} for model-based reinforcement learning
  \newblock  In \emph{2017 {IEEE} International Conference on Robotics and
  Automation ({ICRA})}, 2017, pp. 1714--1721.

\end{thebibliography}

\end{document}